\documentclass[10pt,journal,letterpaper]{IEEEtran}

\usepackage{psfrag}
\usepackage{epsfig}
\usepackage{graphicx}
\usepackage{multicol}
\usepackage{cite} 
\RequirePackage{amsfonts}
\RequirePackage{amssymb}
\RequirePackage{amsopn}
\RequirePackage{amsmath}
\RequirePackage{pifont}
\RequirePackage{mathrsfs}
\RequirePackage{pstricks}

\def\atauxout{\csname @auxout\endcsname}%
\def\labelii#1{\immediate\write\atauxout{%
    \noexpand\newlabel{#1}{{\theenumii}{\thepage}}}}

\newcommand{\vect}[1]{\mathbf{#1}} 
 
\newcommand{\set}[1]{\mathcal{#1}}   

\newcommand{\Hb}{H_{\textnormal{b}}}

\newcommand{\indep}{{\;\bot\!\!\!\!\!\!\bot\;\,}}  
\newcommand{\markov}{\textnormal{\mbox{$\multimap\hspace{-0.73ex}-\hspace{-2ex}-$}}}

\newbox\measurebox %
\newlength{\firstmini} %
\newcommand{\mytextandeps}[2]{%
  \setbox\measurebox\hbox{\epsfig{file=#2}}
  \setlength{\firstmini}{\linewidth}%
  \addtolength{\firstmini}{-\wd\measurebox}%
  \addtolength{\firstmini}{-1em}%
  \begin{minipage}[t]{\firstmini}
    #1
  \end{minipage} \hfill %
  \setbox\measurebox\vbox{\unhbox\measurebox} %
  \setlength{\firstmini}{\ht\measurebox} %
  \addtolength{\firstmini}{\dp\measurebox} %
  \ht\measurebox=0pt \dp\measurebox=\firstmini %
  \box\measurebox
}%
  
\newgray{mygray}{0.9}
\makeatletter
  {\egroup     
  \noindent
  \centerline{
  \psshadowbox[shadowsize=0.3em,framesep=1.5em, fillstyle=solid, fillcolor=mygray]
              {\box\@tempboxa}
              }
  \par
  \vskip 0pt plus1ex 
  }%
\makeatother

\makeatletter
  {\egroup     
  \noindent
  \centerline{
  \psshadowbox[shadowsize=0.3em,framesep=1.5em, fillstyle=solid, fillcolor=white]
              {\box\@tempboxa}
              }
  \par
  \vskip 0pt plus1ex 
  }%

\newtheorem{theorem}{Theorem}
\newtheorem{lemma}{Lemma}

\newtheorem{proposition}{Proposition}

\title{The State-Dependent Semideterministic \\Broadcast Channel}

\author{Amos~Lapidoth,~\IEEEmembership{Fellow,~IEEE}~and~Ligong~Wang,~\IEEEmembership{Member,~IEEE}
\thanks{The material in this paper was presented in part at IEEE 2012
  International Symposium on Information Theory, Cambridge, MA, USA,
  1--6 Jul. 2012. L.W. is supported by the US Air Force Office of 
Scientific Research under Grant No.  FA9550-11-1-0183, and by the
National Science Foundation under Grant No. CCF-1017772.}
\thanks{A. L. is with ETH Zurich, ETF E107, Sternwartstrasse 7,
    Z\"urich 
    8092, Switzerland. E-mail:
    \texttt{lapidoth@isi.ee.ethz.ch}.}
\thanks{L. W. is with the Massachusetts Institute of Technology, 77
  Massachusetts Avenue, 36-687, Cambridge, MA 02139,  
    USA. E-mail: \texttt{wlg@mit.edu}.}
}  

\date{}

\begin{document}

\maketitle

\begin{abstract}
  We derive the capacity region of the state-dependent
  semideterministic broadcast channel with noncausal state-information
  at the transmitter. One of the two outputs of this channel is
  a deterministic function of the channel input and the channel state,
  and the state is assumed to be known noncausally to the transmitter
  but not to the receivers. We show that appending the state to the deterministic
  output does not increase capacity. 

  We also derive an outer bound on the capacity of general (not
  necessarily semideterministic) state-dependent broadcast channels.
\end{abstract}


\section{Introduction}

\IEEEPARstart{W}{e} characterize the capacity region of the discrete,
memoryless, 
state-dependent, semideterministic broadcast channel. This channel has
a single transmitting node, two receiving nodes, and an internal
state, all of which are assumed to take value in finite sets. One of
the receiving nodes---the ``deterministic receiver''---observes a
symbol~$Y$ that is a deterministic function of the transmitted symbol
$x$ and the (random) state~$S$
\begin{subequations}\label{eq:channel}
\begin{equation}
  Y = f(x,S)\quad\textnormal{with probability one},\label{eq:channel1}
\end{equation}
and the other receiving node---the ``nondeterministic
receiver''---observes a symbol $Z$, which is random: conditional on
the input being~$x$ and the state being~$s$, the probability that it
equals $z$ is $W(z|x,s)$:
\begin{equation}
  \Pr[Z=z|X=x, S=s] = W(z|x,s). \label{eq:channel2}
\end{equation}
The state sequence $\vect{S}$ is assumed to be independent and
identically distributed (IID) according to some law $P_{S}(\cdot)$
\begin{equation}
  \label{eq:channel3}
  \Pr[S=s] = P_{S}(s)
\end{equation}
\end{subequations}
and to be revealed to the encoder in a noncausal way: all future
values of the state are revealed to the transmitter before
transmission begins.

We consider a scenario where the encoder wishes to convey \emph{two
  private messages}: $M_y \in \{1,\ldots,2^{nR_y}\}$ to the
deterministic receiver, and $M_z \in \{1,\ldots,2^{nR_z}\}$ to the
nondeterministic receiver, where $R_y$ and $R_z$ denote the
\emph{rates} (in bits per channel use) of data transmission to the
deterministic and nondeterministic receivers.\footnote{To be precise,
  we should replace $2^{nR_{y}}$ and $2^{n R_{z}}$ with their integer
  parts, but, for typographical reasons, we shall not.} The
messages $M_y$ and $M_z$ are assumed to be independent and uniformly
distributed. As for the broadcast channel without a state
\cite{coverthomas91,elgamalkim11}, we define the \emph{capacity
  region} of this channel as the closure of all rate-pairs that are
achievable in the sense that the probability that at least one of the
receivers decodes its message incorrectly can be made arbitrarily
close to zero.

The main result of this paper is a single-letter characterization of
the capacity region:
\begin{theorem}\label{thm:main}
  The capacity region of the channel \eqref{eq:channel} when the
  states are known noncausally to the transmitter is the convex
  closure of the union of rate-pairs $(R_y,R_z)$ satisfying
  \begin{subequations}\label{eq:main}
  \begin{IEEEeqnarray}{rCl}
    R_y & < & H(Y|S)\\
    R_z & < & I(U;Z)-I(U;S)\\
    R_y+R_z & < & H(Y|S)+I(U;Z) - I(U;S,Y)
  \end{IEEEeqnarray}
  \end{subequations}
  over all joint distribution on $(X,Y,Z,S,U)$ whose marginal on $S$
  is the 
  given state distribution $P_{S}$ and under which, conditional on $X$
  and $S$, the channel outputs $Y$ and $Z$ are drawn according to the
  channel law \eqref{eq:channel} independently of $U$:
  \begin{IEEEeqnarray}{rCl}
    \lefteqn{P_{XYZSU}(x,y, z,s,u)}~~~~~~~\nonumber \\ 
    & = & P_{S}(s) \, P_{XU|S}(x,u|s) \, \mathbf{1}\bigl\{y =
      f(x,s)\bigr\} \, W(z|x,s).\label{eq:distribution}
      \IEEEeqnarraynumspace 
  \end{IEEEeqnarray}
  Here $\mathbf{1}\{\cdot\}$ denotes the indicator
  function.\footnote{The value of $\mathbf{1}\{\text{statement}\}$ is
    $1$ if the statement is true and is $0$ otherwise.}  Moreover,
  this is also the capacity region when the state sequence is also
  revealed to the deterministic receiver, i.e., when the mapping
  $f(\cdot,\cdot)$ is replaced by the mapping $(x,s) \mapsto \bigl(
  f(x,s), s \bigr)$.
\end{theorem}

\begin{IEEEproof}
  See Sections~\ref{sec:direct} and~\ref{sec:converse}.
\end{IEEEproof}

As to the cardinality of the auxiliary random variable $U$:
\begin{proposition}\label{prp:cardinality}
  To exhaust the capacity region of the channel~\eqref{eq:channel}, 
  we may restrict the auxiliary random variable $U$ in
  \eqref{eq:main} to take value in a set $\set{U}$ whose cardinality
  $|\set{U}|$ is bounded by 
  \begin{equation}
    |\set{U}| \le |\set{X}|\cdot|\set{S}|+1,
  \end{equation}
  where $\set{X}$ and $\set{S}$ denote the input and state alphabets.
\end{proposition}

\begin{IEEEproof}
  See Appendix~\ref{app:cardinality}.
\end{IEEEproof}

Broadcast channels \emph{without states} have been studied extensively
\cite{cover98}. Our work can be considered as an extension to
broadcast channels with states of prior work by Gel'fand, Marton, and
Pinsker on deterministic and semideterministic broadcast channels
without states
\cite{gelfand77,marton77,pinsker78,marton79,gelfandpinsker80,elgamalkim11}.
State-dependent broadcast channels were also considered before
\cite{steinberg05,steinbergshamai05,khosravimarvasti11}, but capacity
regions of most such channels are still unknown.

Steinberg \cite{steinberg05} 
studied the \emph{degraded} state-dependent broadcast channel with
causal and with noncausal state-information at the transmitter. He
derived the capacity region for the causal case, but for the noncausal
case his outer and inner bounds 
do not coincide. Steinberg and Shamai~\cite{steinbergshamai05} then
derived an inner bound for general (not necessarily degraded)
state-dependent broadcast channels with noncausal
state-information. This inner bound is based on Marton's inner bound
for broadcast channels without states \cite{marton79} and on
Gel'fand-Pinsker coding \cite{gelfandpinsker80_3}. In fact, the direct
part of our Theorem~\ref{thm:main} can be deduced from
\cite{steinbergshamai05} with a proper choice of the auxiliary random
variables (see Section \ref{sub:direct1}). 

Our proof of the converse part of Theorem~\ref{thm:main} borrows from
the Gel'fand-Pinsker converse for single-user channels with states
\cite{gelfandpinsker80_3} as well as from the K\"orner-Marton
\cite{marton79} and the Nair-El~Gamal \cite{nairelgamal07} approaches
to outer-bounding the capacity region of broadcast channels without
states. But it also has a new element: \emph{the choice/definition of
  the auxiliary random variable depends on the codebook.}
As we demonstrate in Section~\ref{sec:general}, our proof can be
extended to general (not necessarily semideterministic)
state-dependent broadcast channels.

Some special cases of Theorem~\ref{thm:main} were solved by 
Khosravi-Farsani and Marvasti \cite{khosravimarvasti11}: the
\emph{fully} deterministic case, the case where the states are known
to the nondeterministic receiver, and the case where the
channel is degraded so $(X,S)\markov Y \markov Z$ forms a Markov
chain. 

The rest of this paper is organized as follows. We prove the direct and
converse parts of Theorem~\ref{thm:main} in Sections~\ref{sec:direct}
and~\ref{sec:converse}. In Section~\ref{sec:example} we apply
Theorem~\ref{thm:main} to a specific channel whose nondeterministic
output is unaffected by the state. Even so, noncausal
state-information is strictly better than causal. We finally derive a
new 
outer bound on general state-dependent broadcast channels in
Section~\ref{sec:general}.

\section{Direct Part}\label{sec:direct}

In this section we prove the direct part of
Theorem~\ref{thm:main}. One way to do this is to use
\cite[Theorem~1]{steinbergshamai05} with the choice of the auxiliary
random variables that we propose in Section~\ref{sub:direct1}\@. For
completeness and simplicity, we also provide a self-contained proof in
Section~\ref{sub:direct2}.

\subsection{Proof based on~\cite{steinbergshamai05}} 
\label{sub:direct1} 
It was shown in
\cite[Theorem 1]{steinbergshamai05} that the capacity region of a
general (not necessarily semideterministic) state-dependent broadcast
channel with noncausal state-information at the transmitter contains
the convex closure of the union of rate-pairs $(R_y,R_z)$ satisfying
\begin{subequations}\label{eq:steinbergshamai05}
\begin{IEEEeqnarray}{rCl}
  R_y & \le & I(U_0,U_y;Y)-I(U_0,U_y;S)\\
  R_z & \le & I(U_0,U_z;Z)-I(U_0,U_z;S)\\
  R_y+R_z & \le & - \bigl[ \max \{ I(U_0;Y),I(U_0;Z) \} - I(U_0;S)
  \bigr]^+ \nonumber\\  
  && {} + I(U_0,U_y;Y) -I(U_0,U_y;S) + I(U_0,U_z;Z)
  \nonumber\\   
  && {} - I(U_0,U_z;S) - I(U_y;U_z|U_0,S),
\end{IEEEeqnarray}
\end{subequations}
where the union is over all joint distribution
on $(X,Y,Z,S,U_0,U_y,U_z)$ whose marginal is $P_{S}$;
that satisfies the Markov condition
\begin{equation}
\label{eq:StShMar}
(U_0,U_y,U_z)\markov (X,S)\markov(Y,Z);
\end{equation}
and under which the conditional law of $(Y,Z)$ given $(X,S)$ is that of
the given channel. 

For the semideterministic channel, we choose the auxiliary random
variables in \eqref{eq:steinbergshamai05} as follows:
\begin{subequations}
\begin{IEEEeqnarray}{rCl}
  U_0 & = & 0\quad\textnormal{(deterministic)}\\
  U_y & = & Y\\
  U_z & = & U.
\end{IEEEeqnarray}
\end{subequations}
Note that the Markov condition~\eqref{eq:StShMar} is satisfied because
$Y$ is a deterministic function of $(X,S)$ and because in
Theorem~\ref{thm:main} we restrict $U$ to be such that $U \markov (X,S)
\markov (Y,Z)$. With this choice of
$U_{0}$, $U_{y}$, and $U_{z}$, \eqref{eq:steinbergshamai05} reduces to
\eqref{eq:main}.

\subsection{Self-contained proof} \label{sub:direct2}

We next provide a self-contained proof of the direct part of
Theorem~\ref{thm:main}. As in~\cite[Theorem 1]{steinbergshamai05}, our
proof is based on Marton's inner bound for general broadcast channels
\cite{marton79,elgamalvandermeulen81} and on Gel'fand-Pinsker coding
\cite{gelfandpinsker80_3}.

First note that the joint distribution \eqref{eq:distribution} can
also be written as 
\begin{IEEEeqnarray}{rCl}
  \lefteqn{P_{XYZSU}(x,y,z,s,u)}~~~\nonumber \\  
  & = & P_S(s) \, P_{YU|S}(y,u|s) \, P_{X|YSU}(x|y,s,u) \,
  W(z|x,s)\label{eq:distribution1} \IEEEeqnarraynumspace
\end{IEEEeqnarray}
with the additional requirement that
\begin{equation}\label{eq:requirey}
  y = f(x,s).
\end{equation}
Further note that, when $P_{YSU}$ is fixed, all the terms on the
right-hand side (RHS)
of \eqref{eq:main} are fixed except for $I(U;Z)$, which is convex in
$P_{X|YUS}$. Since $I(U;Z)$ only appears
with a positive sign on the RHS of \eqref{eq:main}, it follows that the
union over all joint distributions of the form~\eqref{eq:main} can be replaced
by a union only over those where $x$ is a deterministic function of
$(y,u,s)$, i.e., of the form
\begin{IEEEeqnarray}{rCl}
  \lefteqn{P_{XYZSU}(x,y,z,s,u)}~~~~ \nonumber \\    & = &
  P_S(s) \,
  P_{YU|S}(y,u|s) \, \mathbf{1}\bigl\{x=g(y,u,s)\bigr\} \, 
  W(z|x,s) \label{eq:distribution2} \IEEEeqnarraynumspace
\end{IEEEeqnarray}
for some $g\colon (y,u,s)\mapsto x$ (and subject to~\eqref{eq:requirey}).
We shall thus only establish the achievability of rate pairs that
satisfy~\eqref{eq:main} for some distribution of the form
\eqref{eq:distribution2}.

Choose a stochastic kernel $P_{YU|S}$ and a mapping $g\colon
(y,u,s)\mapsto x$ which, combined with $P_S$ and the channel law,
determines the joint distribution \eqref{eq:distribution2} for which
\eqref{eq:requirey} is satisfied.
For a given block-length $n$, we construct a random
code as follows: 

\textbf{Codebook:}
  Generate
  $2^{nR_y}$ $y$-bins, each containing 
$2^{n\tilde{R}_y}$ $y$-tuples where the $l_{y}$-th $y$-tuple in the
$m_{y}$-th bin 
\begin{equation*}
  \vect{y}(m_y,l_y), \quad 
m_y\in\{1,\ldots,2^{nR_y}\}, \; l_y\in\{1,\ldots,2^{n\tilde{R}_y}\}
\end{equation*}
is generated IID according to $P_Y$ (the $Y$-marginal
of~\eqref{eq:distribution2}) independently of the other
$y$-tuples. Additionally, 
generate $2^{nR_z}$ $u$-bins, each containing $2^{n\tilde{R}_z}$
$u$-tuples, where the $l_{z}$-th $u$-tuple in the $m_{z}$-th $u$-bin
\begin{equation*}
  \vect{u}(m_z,l_z), \quad m_z \in \{1,\ldots,2^{nR_z}\}, \; l_z\in \{1,
  \ldots, 2^{n\tilde{R}_z}\} 
\end{equation*}
is drawn IID
according to $P_U$ (the $U$-marginal of
\eqref{eq:distribution2}) independently of the other $u$-tuples
and of the $y$-tuples.

\textbf{Encoder:} To send Message~$m_y\in\{1,\ldots,2^{n R_y}\}$ to
  the deterministic receiver and Message~$m_z\in\{1,\ldots,2^{n R_z}\}$
  to the nondeterministic receiver, 
  look for a $y$-tuple $\vect{y}(m_y,l_y)$ in
  $y$-bin~$m_{y}$ and a $u$-tuple $\vect{u}(m_z,l_z)$ in $u$-bin~$m_{z}$
  such that $\bigl( \vect{y}(m_y,l_y), \vect{u}(m_z,l_z) \bigr)$ is
  jointly typical with the state sequence $\vect{s}$:
  \begin{equation}\label{eq:jointtyp}
    \bigl( \vect{y}(m_y,l_y), \vect{u}(m_z,l_z), \vect{s} \bigr) \in
    \set{T}_\epsilon^{(n)}\left(P_{YUS}\right),
  \end{equation}
  where $\set{T}_\epsilon^{(n)}(\cdot)$ denotes the
  \emph{$\epsilon$-strongly 
  typical set} with respect to a certain distribution. If such a pair
  can be found, send
  \begin{equation}
    \vect{x}=g\bigl(\vect{y}(m_y,l_y), \vect{u}(m_z,l_z), \vect{s} \bigr),
  \end{equation}
  where in the above $g(\vect{y}, \vect{u},\vect{s})$ denotes the
  application of the function $g(y,u,s)$ componentwise. (Note that in
  this case 
  the sequence received by the deterministic receiver will be
  $\vect{y}(m_y,l_y)$.) Otherwise send an arbitrary codeword.

\textbf{Deterministic decoder:} Try to find the \emph{unique}
  $y$-bin, say $m_y'$, 
  that contains the received sequence $\vect{y}$ and output its
  number $m_y'$. If there is more than one such bin, declare an error.

\textbf{Nondeterministic decoder:} Try to find the
  \emph{unique} $u$-bin $m_z'$ which contains a 
  $\vect{u}(m_z',l_z')$ that is jointly typical with the received
  sequence $\vect{z}$:
  \begin{equation}
    \bigl(\vect{u}(m_z',l_z'), \vect{z}\bigr) \in
    \set{T}_{2\epsilon}^{(n)}\left(P_{UZ}\right),
  \end{equation}
  and output $m_z'$. If more than
  one or no such bin can be found, declare an error.

We next analyze the error probability of the above coding
scheme. There are three types of errors:

  \textbf{Encoder errs.} This happens only if there is no pair
    $(l_y,l_z)\in 
    \{1,\ldots,2^{n\tilde{R}_y}\} \times
    \{1,\ldots,2^{n\tilde{R}_z}\}$  that satisfies
    \eqref{eq:jointtyp}. To bound this 
    probability, we use the Multivariate
    Covering Lemma \cite[Lemma 8.2]{elgamalkim11}, which we restate as
    follows:
    \begin{lemma}\label{lem:mcl}
      Fix some joint distribution $P_{A_{(0)}\ldots A_{(k)}}$ on
      $(A_{(0)},\ldots,A_{(k)})$, and fix positive $\tilde{\epsilon}$ and
      $\epsilon$ 
      with $\tilde{\epsilon}<\epsilon$. Let $A_{(0)}^n$ be a random
      sequence satisfying
      \begin{equation}
        \lim_{n\to\infty} \Pr \left[ A_{(0)}^n \in
          \set{T}_{\tilde{\epsilon}}^{(n)} (P_{A_{(0)}}) \right] 
        = 1.
      \end{equation}
      For each $j\in\{1,\ldots,k\}$, let $A_{(j)}^n(m_j)$,
      $m_j\in\{1,\ldots,2^{nr_j}\}$, be pairwise
      independent conditional on $A_{(0)}^n$, each distributed according
      to 
      $\prod_{i=1}^n P_{A_{(j)}|A_{(0)}=a_{(0),i}}$. Assume that
      \begin{equation*}
      \left\{A_{(j)}^n(m_j)\colon
        m_j\in\{1,\ldots,2^{nr_j}\}\right\},\quad j \in\{1,\ldots, k\}
      \end{equation*}
      are mutually independent
      conditional on $A_{(0)}^n$. Then there exists $\delta(\epsilon)$
      which tends to zero as $\epsilon$ tends to zero such that
      \begin{equation}
        \lim_{n\to\infty} \Pr \left[\substack{\displaystyle 
          (A_{(0)}^n,A_{(1)}^n(m_1),\ldots,A_{(k)}^n(m_k)) \notin
          \set{T}_\epsilon^{(n)}\\ \displaystyle 
        ~~~~~~~~~~~~~~~~~~\textnormal{ for all }(m_1,\ldots,m_k)}
        \right]  =  0
      \end{equation}
      provided that, for all $\set{J}\subseteq \{1,\ldots,k\}$ with
      $|\set{J}| \ge 2$,
      \begin{equation}
        \sum_{j\in\set{J}} r_j > \sum_{j\in\set{J}} H(A_{(j)}|A_{(0)}) -
        H(\{A_{(j)}\colon j\in\set{J}\} | A_{(0)}) + \delta(\epsilon),
      \end{equation}
      where the conditional entropies are computed with respect to
      $P_{A_{(0)}\ldots A_{(k)}}$.
    \end{lemma}

    We apply Lemma~\ref{lem:mcl} by choosing $k=3$, $A_{(0)}=0$
    (deterministic) so $\tilde{\epsilon}=0$, and
    \begin{subequations}
    \begin{IEEEeqnarray}{rCl}
      A_{(1)} & = & Y,\quad r_1 = \tilde{R}_y,\\
      A_{(2)} & = & U,\quad r_2 = \tilde{R}_z,\\
      A_{(3)} & = & S,\quad r_3 = 0.
    \end{IEEEeqnarray}
    \end{subequations}
    The joint distribution is chosen to be $P_{YUS}$.  We then obtain
    that the probability that the encoder errs tends to zero as $n$
    tends to infinity provided that
    \begin{subequations}\label{eq:Rtilde}
    \begin{IEEEeqnarray}{rCl}
      \tilde{R}_y & > & I(Y;S)+\delta(\epsilon)\\
      \tilde{R}_z & > & I(U;S)+\delta(\epsilon)\\
      \tilde{R}_y+\tilde{R}_z & > &
      H(Y)+H(U)+H(S)\nonumber\\
      & & {} -H(Y,U,S) + \delta(\epsilon). \IEEEeqnarraynumspace
    \end{IEEEeqnarray}
    \end{subequations}

  \textbf{Deterministic decoder errs.} This happens only if
    there is more 
    than one bin that contains the received 
    $\vect{y}$. We may now assume that the encoding was successful so
    \eqref{eq:jointtyp} is satisfied. Then
    $\vect{y}$ is in $\set{T}_\epsilon^{(n)}(P_Y)$, and
    \begin{equation}\label{eq:HY}
      P_Y(\vect{y}) \le 2^{-n(H(Y)-\delta(\epsilon))}
    \end{equation}
    where $\delta(\epsilon)$ tends to zero when $\epsilon$ tends to
    zero. Hence the probability that a specific $y$-tuple in a bin
    that was not 
    chosen by the encoder, which, by our code construction, was
    independently chosen from the received $\vect{y}$, happens to be
    the same as $\vect{y}$, is upper-bounded by the RHS of
    \eqref{eq:HY}. Further note that the total number of $y$-tuples
    outside the bin chosen by the encoder is
    $2^{n\tilde{R}_y}\left(2^{nR_y}-1\right)$. Using the union bound,
    we obtain that the probability that the deterministic decoder errs is
    at most 
    \begin{equation}
      2^{n\tilde{R}_y}\left(2^{nR_y}-1\right)
      2^{-n(H(Y)-\delta(\epsilon))},
    \end{equation} 
    which tends to zero as $n$ tends to 
    infinity provided that
    \begin{equation}
      R_y+\tilde{R}_y < H(Y)-\delta(\epsilon). \label{eq:R1}
    \end{equation}

  \textbf{Nondeterministic decoder errs.} This happens if
    either the $u$-tuple 
    $\vect{u}(m_z,l_z)$ is not jointly typical with the received
    $z$-tuple, or if a $u$-tuple in a different bin
    happens to be
    jointly typical with the received $z$-tuple. Assuming that the encoding
    was successful, the probability of
    the former case tends to zero as $n$ tends to infinity by
    \eqref{eq:jointtyp} and by the
    Markov Lemma~\cite[Lemma~12.1]{elgamalkim11}. To upper-bound the
    probability of the latter case, note that any
    $\vect{u}(m_z',l_z')$, where $m_z'\neq m_z$, is
    chosen independently of 
    $\vect{u}(m_z,l_z)$ and $\vect{y}(m_y,l_y)$, and is hence also
    independent of the received $\vect{z}$. By the Joint Typicality
    Lemma \cite[p.29]{elgamalkim11} we have
    \begin{equation}
      \Pr \left[ \left(\vect{U}(m_z',l_z'), \vect{Z} \right)
      \in \set{T}_{2\epsilon}^{(n)}(P_{UZ}) \right] \le
    2^{-n\left(I(U;Z)-\delta(\epsilon)\right)}
    \end{equation}
    where the probability is computed with respect to the randomly
    chosen codebook. Next note that the total number of such
    $u$-tuples is $2^{n\tilde{R}_z}\left(2^{n R_z}-1\right)$. 
    Applying the union bound, we obtain that the probability that
    there exists at least one $u$-tuple that is not in the chosen bin
    but that is jointly typical with $\vect{z}$ is at most 
    \begin{equation}
      2^{n\tilde{R}_z}\left(2^{n R_z}-1\right)
      2^{-n\left(I(U;Z)-\delta(\epsilon)\right)}, 
    \end{equation}
    which tends to zero
    as $n$ tends to infinity provided that
    \begin{equation}
      R_z+\tilde{R}_z < I(U;Z) - \delta(\epsilon). \label{eq:R2}
    \end{equation}

Summarizing \eqref{eq:Rtilde}, \eqref{eq:R1}, and \eqref{eq:R2}, and
letting $\epsilon$ tend to zero, we
conclude that the above coding scheme has vanishing error
probability as $n$ tends to infinity for all $(R_y,R_z)$
satisfying~\eqref{eq:main}. By time-sharing we further achieve the 
convex hull of all rate-pairs satisfying \eqref{eq:main} for joint
distributions of the form \eqref{eq:distribution2}.
This concludes the proof of the direct
part of Theorem~\ref{thm:main}.

\section{Converse Part}\label{sec:converse}


In this section we show that, even if the state sequence $\vect{S}$ is
revealed 
to the deterministic receiver (which observes $\vect{Y}$), any achievable
rate-pair must be in the convex closure of the union of rate-pairs
satisfying \eqref{eq:main}. 

Given any code of block-length $n$, we
first derive a bound on $R_y$:
\begin{IEEEeqnarray}{rCl}
  nR_y & = & H(M_y) \\
  & \le & I(M_y; Y^n,S^n) + n\epsilon_n \label{eq:R1_y}\\
  & = & I(M_y;Y^n|S^n) + n\epsilon_n \label{eq:R1_z}\\
  & = & \sum_{i=1}^n I(M_y;Y_i|Y^{i-1},S^n)
  +n\epsilon_n \label{eq:R1_3}\\ 
  & \le & \sum_{i=1}^n H(Y_i|Y^{i-1},S^n) +
  n\epsilon_n \label{eq:R1_4}\\ 
  & \le & \sum_{i=1}^n H(Y_i|S_i) + n\epsilon_n,\label{eq:R1_5}
\end{IEEEeqnarray}
where $\epsilon_n$ tends to zero as $n$ tends to infinity. Here,
\eqref{eq:R1_y} follows from Fano's Inequality; \eqref{eq:R1_z}
because $M_y$ and $S^n$ are independent; \eqref{eq:R1_3} from the
chain rule; \eqref{eq:R1_4} by dropping negative terms; and
\eqref{eq:R1_5} because conditioning cannot increase entropy.

We next bound $R_z$ as in 
\cite{gelfandpinsker80_3}: 
\begin{IEEEeqnarray}{rCl}
  nR_z & = & H(M_z)\\
  & \le & I(M_z;Z^n) + n\epsilon_n \label{eq:R2_y}\\
  & = & \sum_{i=1}^n I(M_z;Z_i|Z^{i-1}) +
  n\epsilon_n \label{eq:R2_z}\\
  & = & \sum_{i=1}^n I\left( \left. M_z,S_{i+1}^n; Z_i \right| Z^{i-1}
    \right) \nonumber \\    && {}- \sum_{i=1}^n
    I\left(\left. S_{i+1}^n;Z_i 
      \right| M_z, 
      Z^{i-1}\right) + n\epsilon_n \label{eq:R2_3} \\
  & = & \sum_{i=1}^n I\left( \left. M_z,S_{i+1}^n; Z_i \right| Z^{i-1}
    \right) \nonumber \\
    &&{} - \sum_{i=1}^n I\left( Z^{i-1}; S_i \left| M_z, S_{i+1}^n
      \right. \right) + n\epsilon_n
    \IEEEeqnarraynumspace \label{eq:R2_4} \\ 
  & = & \sum_{i=1}^n I\left( \left. M_z,S_{i+1}^n; Z_i \right| Z^{i-1}
    \right) \nonumber \\   
    && {} - \sum_{i=1}^n I\left(  M_z, Z^{i-1}, S_{i+1}^n ; S_i
    \right) + n\epsilon_n \label{eq:R2_5} \\
  & \le &  \sum_{i=1}^n I\left( M_z, Z^{i-1}, S_{i+1}^n; Z_i \right)
  \nonumber \\    &&{}-
  \sum_{i=1}^n I\left( M_z, Z^{i-1}, S_{i+1}^n ; S_i
    \right) + n\epsilon_n \label{eq:R2_6} \\
  & = & \sum_{i=1}^n I(V_i;Z_i) - I(V_i;S_i) + n\epsilon_n. \label{eq:R2_7}
\end{IEEEeqnarray}
Here, \eqref{eq:R2_y} follows from Fano's Inequality;
\eqref{eq:R2_z} and \eqref{eq:R2_3} from the
chain rule; \eqref{eq:R2_4} from Csisz\'ar's Identity
\cite{csiszarkorner81}
\begin{equation}
  \sum_{i=1}^n I \left( \left. C_{i+1}^n;D_i \right|D^{i-1} \right) =
  \sum_{i=1}^n I\left (D^{i-1};C_i\left|C_{i+1}^n
    \right. \right); \label{eq:csiszar} 
\end{equation}
\eqref{eq:R2_5} because $S_i$ and $(M_z,S_{i+1}^n)$ are independent;
\eqref{eq:R2_6} from the chain rule and by dropping negative terms;
and \eqref{eq:R2_7} by defining the 
auxiliary random variables
\begin{equation}\label{eq:defV}
  V_i \triangleq (M_z,Z^{i-1},S_{i+1}^n), \quad i \in\{1,\ldots,n\}.
\end{equation}

We next bound the sum rate $R_y+R_z$:
\begin{IEEEeqnarray}{rCl}
  n(R_y+R_z) & = & H(M_y,M_z)\\
  & = & H(M_z) +H(M_y|M_z) \\
  & \le & I(M_z; Z^n) + I(M_y;Y^n,S^n|M_z) +
  n\epsilon_n, \IEEEeqnarraynumspace\label{eq:R12_y}
\end{IEEEeqnarray}
where the last step follows from Fano's
Inequality. Of the two mutual informations on the RHS of
\eqref{eq:R12_y} we first bound $I(M_z;Z^n)$:
\begin{IEEEeqnarray}{rCl}
  I(M_z;Z^n) & = & \sum_{i=1}^n I(M_z;Z_i|Z^{i-1}) \label{eq:R12_z1}\\ 
  & \le & \sum_{i=1}^n I(M_z,Z^{i-1}; Z_i) \label{eq:R12_z2}\\
  & = & \sum_{i=1}^n I \left( \left. M_z, Z^{i-1},S_{i+1}^n,
      Y_{i+1}^n; Z_i \right.\right) \nonumber \\    &&{} -
  \sum_{i=1}^n 
  I\left(\left. S_{i+1}^n, Y_{i+1}^n; Z_i \right| M_z, Z^{i-1} \right)
  \label{eq:R12_z3}\\
  & = & \sum_{i=1}^n I \left( \left. M_z, Z^{i-1},S_{i+1}^n,
      Y_{i+1}^n; Z_i \right.\right) \nonumber \\   &&{} -
  \sum_{i=1}^n 
  I\left( Z^{i-1}; S_i,Y_i\left| M_z,S_{i+1}^n,Y_{i+1}^n \right.\right)
  \label{eq:R12_z4} \IEEEeqnarraynumspace \\
  & = & \sum_{i=1}^n I \left( \left. M_z, Z^{i-1},S_{i+1}^n,
      Y_{i+1}^n; Z_i \right.\right) \nonumber \\   &&{} -
  \sum_{i=1}^n 
  I\left( M_z,Z^{i-1},S_{i+1}^n,Y_{i+1}^n; S_i,Y_i \right) \nonumber
  \\   
  & & {} + \sum_{i=1}^n
  I\left(M_z,S_{i+1}^n,Y_{i+1}^n; S_i,Y_i\right). \label{eq:R12_z5}
\end{IEEEeqnarray}
Here, \eqref{eq:R12_z1}, \eqref{eq:R12_z2},
and \eqref{eq:R12_z3} follow from the chain
rule; \eqref{eq:R12_z4} by applying Csisz\'ar's Identity
\eqref{eq:csiszar} between $(S^n,Y^n)$ and $Z^n$; and
\eqref{eq:R12_z5} again from the chain rule. 

We next study the sum of the last term on the RHS of \eqref{eq:R12_z5}
and the second mutual information on the RHS of \eqref{eq:R12_y}:
\begin{IEEEeqnarray}{rCl}
  \lefteqn{\sum_{i=1}^n I\left(M_z,S_{i+1}^n,Y_{i+1}^n;S_i,Y_i\right)
    + I(M_y;Y^n,S^n|M_z)} ~~~~~~~\nonumber \\   
  & = & \sum_{i=1}^n I\left(M_z,S_{i+1}^n,Y_{i+1}^n; S_i,Y_i\right)
  \nonumber \\  
  & & {} +
  \sum_{i=1}^n I\left(M_y;S_i,Y_i\left| M_z,S_{i+1}^n,Y_{i+1}^n
    \right. \right) \label{eq:R12_31}\\
  & = & \sum_{i=1}^n I\left(M_y,M_z,S_{i+1}^n,Y_{i+1}^n ; S_i,Y_i
  \right) \label{eq:R12_32}\\
  & = & \sum_{i=1}^n I\left(M_y,M_z,S_{i+1}^n,Y_{i+1}^n ; S_i,Y_i
  \right) \nonumber \\  
  & & {} + \sum_{i=1}^n I\left(S^{i-1} ; S_i,Y_i \left|
      M_y,M_z,S_{i+1}^n, Y_{i+1}^n \right. \right) \nonumber
  \\   
  & & {} - \sum_{i=1}^n I\left(\left.S_{i+1}^n, Y_{i+1}^n;S_i \right|
    M_y,M_z, S^{i-1} \right) \label{eq:R12_33}\\
  & = & \sum_{i=1}^n I\left(M_y,M_z, S^{i-1},S_{i+1}^n,
    Y_{i+1}^n ; S_i,Y_i\right) \nonumber \\  
  & & {} - \sum_{i=1}^n I\left(\left. S_{i+1}^n, Y_{i+1}^n ; S_i \right|
    M_y,M_z, S^{i-1} \right) \label{eq:R12_34} \\
  & = & \sum_{i=1}^n I\left(M_y,M_z, S^{i-1},S_{i+1}^n,
    Y_{i+1}^n ; S_i,Y_i\right) \nonumber \\  
  & & {} - \sum_{i=1}^n I \left(
    M_y,M_z,S^{i-1},S_{i+1}^n,Y_{i+1}^n ; S_i \right) \label{eq:R12_35} \\
  & = & \sum_{i=1}^n I \left( \left. M_y,M_z,S^{i-1},S_{i+1}^n,
      Y_{i+1}^n; Y_i \right| S_i \right) \label{eq:R12_36}\\
  & = & \sum_{i=1}^n H(Y_i|S_i). \label{eq:R12_37}
\end{IEEEeqnarray}
Here, \eqref{eq:R12_31} and \eqref{eq:R12_32} follow from the chain
rule; \eqref{eq:R12_33} by applying Csisz\'ar's Identity between
$(S^n, Y^n)$ and $S^n$; 
\eqref{eq:R12_34}~from the chain rule; \eqref{eq:R12_35} because $S_i$
and $(M_y,M_z,S^{i-1})$ are independent; \eqref{eq:R12_36} again from the
chain rule; and \eqref{eq:R12_37} because, given $(M_y,M_z,S^n)$, the
channel inputs $X^n$ are determined by the encoder, and hence $Y^n$
are also determined, so
\begin{equation}
  H\left(Y_i\left|M_y,M_z,S^n,Y_{i+1}^n\right.\right)=0.
\end{equation} 

Combining \eqref{eq:R12_y}, \eqref{eq:R12_z5}, and
\eqref{eq:R12_37}, using the definitions \eqref{eq:defV}, and further
defining 
\begin{equation}\label{eq:defT}
  T_i \triangleq Y_{i+1}^n,\quad i\in\{1,\ldots,n\},
\end{equation}
we obtain
\begin{IEEEeqnarray}{rCl}
  n(R_y+R_z) & \le &  \sum_{i=1}^n I (V_i,T_i; Z_i)-\sum_{i=1}^n
  I(V_i,T_i;S_i,Y_i ) \nonumber \\    &&{}+
  \sum_{i=1}^n H(Y_i|S_i) + n \epsilon_n. \label{eq:R12_4}
\end{IEEEeqnarray}

Summarizing \eqref{eq:R1_5}, \eqref{eq:R2_7}, and \eqref{eq:R12_4} and
letting $n$ tend to infinity, we
obtain that any achievable rate-pair $(R_y,R_z)$ must be
contained in the convex closure of the union of rate-pairs satisfying
\begin{subequations}\label{eq:almost}
\begin{IEEEeqnarray}{rCl}
    R_y & < & H(Y|S)\\
    R_z & < & I(V;Z)-I(V;S)\\
    R_y+R_z & < & H(Y|S)+I(V,T;Z) - I(V,T;S,Y) \IEEEeqnarraynumspace
\end{IEEEeqnarray}
\end{subequations}
where, given $(X,S)$, the outputs $(Y,Z)$ are drawn according to the
channel law \eqref{eq:channel} independently of the auxiliary
random variables $(V,T)$. 

To prove the converse part of
Theorem~\ref{thm:main}, it remains to replace $V$ and $T$ with a
single auxiliary random variable. I.e., it remains to find an
auxiliary random variable $U$ such that
\begin{subequations}\label{eq:strange}
  \begin{equation}
    \label{eq:strange1}
    I(V;Z)-I(V;S) \le I(U;Z)-I(U;S)
  \end{equation}
  and
  \begin{multline}
    \label{eq:strange2}
    H(Y|S)+I(V,T;Z)-I(V,T;S,Y) \le  \\ H(Y|S) + I(U;Z) - I(U;S,Y).
  \end{multline}
\end{subequations}
In fact, as we shall see, either choosing $U$ to be $V$ will satisfy
\eqref{eq:strange} or else choosing it to be $(V,T)$
will satisfy \eqref{eq:strange}. If we choose $U=V$, then
\eqref{eq:strange1} is satisfied with equality, and
the requirement \eqref{eq:strange2} becomes
\begin{equation}\label{eq:strange3}
  I(T;Z|V)-I(T;S,Y|V) \le 0.
\end{equation}
On the other hand, if we choose $U=(V,T)$, then \eqref{eq:strange2} is
satisfied with equality, and the requirement \eqref{eq:strange1} becomes
\begin{equation}\label{eq:strange4}
  I(T;Z|V)-I(T;S|V) \ge 0.
\end{equation}
It remains to show that \emph{at least one} of the two requirements
\eqref{eq:strange3} 
and \eqref{eq:strange4} must be satisfied: if it is
\eqref{eq:strange3}, then we shall choose $U$ as $V$, and if it is
\eqref{eq:strange4}, then we shall choose $U$ as $(V,T)$.
To this end we note that for all random variables $T,Z,V,S,Y$
\begin{equation}
  I(T;Z|V)-I(T;S,Y|V) \le I(T;Z|V)-I(T;S|V),
\end{equation}
because the RHS minus the left-hand side is $I(T;Y|S,V)$, which is
nonnegative. This implies that at least one of \eqref{eq:strange3}
and \eqref{eq:strange4} must hold. We have thus shown that there
must exist a $U$ which satisfies both inequalities in
\eqref{eq:strange}, hence 
the bounds \eqref{eq:almost} can be relaxed to \eqref{eq:main}. This
concludes the proof of the converse part of Theorem~\ref{thm:main}.

\section{An Example}\label{sec:example}

Consider a broadcast channel whose input, output, and
state alphabets are all binary and whose law is
\begin{subequations}\label{eq:example}
\begin{IEEEeqnarray}{rCl}
  P_S(1) & = & 1-P_S(0) = \sigma\\
  Y & = & x\oplus S\\
  W(Z=x|x,s) & = & 1- W(Z=x\oplus 1|x,s)= 1-p \IEEEeqnarraynumspace
\end{IEEEeqnarray}
\end{subequations}
for some constants $0 \leq p,\sigma \leq 1$.  The deterministic
output~$Y$ of this channel is the modulo-two sum of the input $x$ and
the state $S$, and the channel from $x$ to the nondeterministic
output~$Z$ is unaffected by the state and is a binary symmetric
channel with crossover probability~$p$.

To cancel the state's effect, the encoder could flip the input~$x$ whenever
$S=1$, but this would hurt the nondeterministic receiver. In fact, 
if the state is unbiased ($\sigma = 0.5$), and if \emph{only causal}
state-information is available at the encoder,\footnote{By ``causal''
  we mean that the transmitter, when transmitting $X_i$, knows the
  past and present states $S^i$ but not the future states
  $S_{i+1}^n$.} then one cannot do better than time-sharing:

\begin{proposition}\label{prp:causal}
  The capacity region of the channel \eqref{eq:example} with
  $\sigma=0.5$ when the 
  states are known \emph{causally} to the transmitter but not to the
  receivers, is the union over $\lambda \in [0,1]$ of rate-pairs
  $(R_y,R_z)$ satisfying
  \begin{subequations}\label{eq:causal}
  \begin{IEEEeqnarray}{rCl}
    R_y & \le & \lambda\\
    R_z & \le & (1-\lambda) \bigl( 1-\Hb(p) \bigr).
  \end{IEEEeqnarray}
  \end{subequations}
  I.e., it is the collection of rate pairs satisfying
  \begin{equation}
    R_y + \frac{R_z}{1-\Hb(p)} \leq 1.
  \end{equation}
\end{proposition}
\begin{proof}
  See Appendix~\ref{app:example}.
\end{proof}

However, with \emph{noncausal} state-information the
transmitter can cancel the effect of the state without hurting the
nondeterministic receiver:
\begin{proposition}\label{prop:example}
  The capacity region of the channel \eqref{eq:example} when the
  states are known noncausally to the transmitter but not to the
  receivers, is the union over $\alpha \in [0,1]$ of rate-pairs
  $(R_y,R_z)$ satisfying
  \begin{subequations} \label{eq:example_region}
  \begin{IEEEeqnarray}{rCl}
    R_y & \le & \Hb(\alpha)\\
    R_z & \le & 1 - \Hb(\beta)
  \end{IEEEeqnarray}
  where
  \begin{equation}\label{eq:beta}
    \beta \triangleq \alpha(1-p) + (1-\alpha)p.
  \end{equation}
  \end{subequations}
\end{proposition}

The capacity regions of the channel \eqref{eq:example} when
$\sigma=0.5$ and $p=0.2$ with noncausal and with causal
state-information are depicted in Figure~\ref{fig:example}.

    \begin{figure}[tbp]
        \centering
        \psfrag{Ry}{\scriptsize $R_y$}
        \psfrag{Rz}{\scriptsize $R_z$}
        \psfrag{noncausal}{\scriptsize Noncausal} 
        \psfrag{causal}{\scriptsize Causal}
        \includegraphics[width=0.5\textwidth]{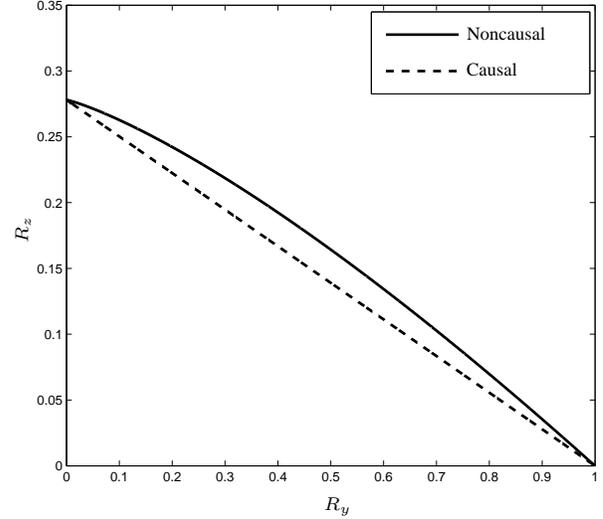}
        \caption{The capacity regions of the channel
          \eqref{eq:example} when $\sigma=0.5$ and $p=0.2$ with
          noncausal (solid line) and with
          causal (dashed line)  
          state-information at the transmitter. 
        } \label{fig:example}   
    \end{figure}

    We present two different proofs for
    Proposition~\ref{prop:example}: the first is based on the
    achievability part of Theorem~\ref{thm:main}; the second is based
    on the fact that revealing the states to the deterministic
    receiver does not increase the capacity region.

\begin{IEEEproof}[First proof of Proposition~\ref{prop:example}]
  We let $U$ be a uniform binary random variable that is
  independent of $S$, and let $X$ be the outcome of feeding $U$
  into a binary symmetric channel of crossover probability $\alpha$
  (independently of $S$). Note that now the channel from $U$ to~$Z$ is
  a binary symmetric channel with crossover probability $\beta$ as
  defined in \eqref{eq:beta}. Using Theorem~\ref{thm:main} we obtain that
  the capacity region contains all rate-pairs $(R_y,R_z)$ satisfying
  \begin{IEEEeqnarray}{rCl}
    R_y & < & H(Y|S) = 1 \label{eq:ex11}\\
    R_z & < & I(U;Z) - I(U;S) \\
    & = & \bigl(1 - \Hb(\beta)\bigr) - 0 \\
    & = & 1 - \Hb(\beta) \label{eq:ex12} \\
    R_y+R_z & < & H(Y|S)+I(U;Z)-I(U;S,Y) \\
    & = & 1 + \bigl(1 - \Hb(\beta)\bigr) - I(U;X) \label{eq:ex14} \\
    & = & 1 + \bigl(1 - \Hb(\beta)\bigr) - \bigl(1 -
    \Hb(\alpha)\bigr) \\
    & = & 1 - \Hb(\beta) + \Hb(\alpha) \label{eq:ex13}
  \end{IEEEeqnarray}
  where \eqref{eq:ex14} follows because $X$ can be computed from $S$
  and~$Y$, and because, given $X$, $U$ is independent of $(S,Y)$.
  Taking the convex
  closure of \eqref{eq:ex11}, \eqref{eq:ex12}, and \eqref{eq:ex13} over
  $\alpha\in[0,1]$, we obtain the 
  region characterized by \eqref{eq:example_region}. 
  
  To see that one cannot do better than \eqref{eq:example_region}, we
  observe that the capacity region of the channel \eqref{eq:example}
  with states known noncausally to the transmitter must be contained
  in the capacity region when the states are also known to both
  receivers. The latter case, however, is equivalent to the following
  broadcast channel without states:
  \begin{subequations}\label{eq:example_no_states}
  \begin{IEEEeqnarray}{rCl}
    y & = & x\\
    W(Z=x|x) & = & 1-W(Z=x\oplus 1|x) = 1-p.
  \end{IEEEeqnarray}
  \end{subequations}
  The capacity region of \eqref{eq:example_no_states} can be found in
  \cite[Example 15.6.5]{coverthomas91} and is the same as the region
  characterized by \eqref{eq:example_region}.
\end{IEEEproof}

\begin{IEEEproof}[Second proof of Proposition~\ref{prop:example}]
  By Theorem~\ref{thm:main}, the capacity region of the channel
  \eqref{eq:example} with states known noncausally to the transmitter
  is unchanged if the states are also revealed to the deterministic
  receiver. When $S$ is revealed to the deterministic receiver, it can
  form $Y \oplus S$ and thus recover $x$. This reduces the channel to
  the one without states \eqref{eq:example_no_states}. Hence the
  capacity region of interest is the same as the capacity region of
  \eqref{eq:example_no_states}, which is given by the union over
  $\alpha\in[0,1]$ of rate-pairs satisfying \eqref{eq:example_region}
  \cite[Example 15.6.5]{coverthomas91}.
\end{IEEEproof}

\section{A General Outer Bound}\label{sec:general}
We next generalize our converse of Section~\ref{sec:converse} to
a broadcast channel that is not necessarily semideterministic. Such a
channel is described by the transition law and the state law
\begin{subequations}\label{eq:general}
\begin{IEEEeqnarray}{rCl}
    \Pr [Y=y,Z=z|X=x,S=s] & = & W(y,z|x,s) \IEEEeqnarraynumspace \\
    \Pr [S=s] & = & P_S(s).
\end{IEEEeqnarray}
\end{subequations}
We let the state sequence $\vect{S}$ be known noncausally to the
transmitter and also known to the receiver which observes~$Y$.
The capacity region is defined in the same way as for the
semideterministic broadcast channel. In particular, we consider only
two private messages.

Applying the techniques of Section~\ref{sec:converse}, we obtain the
following outer bound on the capacity region of the
channel~\eqref{eq:general}. (The bound is tight for semideterministic
channels.)
\begin{proposition}\label{prp:outer}
  The capacity region of the channel \eqref{eq:general}, with the
  state sequence being revealed noncausally to the transmitter and
  also revealed to the receiver which observes~$Y$, is contained in
  the convex closure of rate-pairs satisfying
  \begin{subequations}
  \begin{IEEEeqnarray}{rCl}
    R_y & < & I(X;Y|S) \\
    R_z & < & I(U;Z) - I(U;S)\\
    R_y+R_z & < & I(X;Y|S)+I(U;Z)-I(U;S,Y)
  \end{IEEEeqnarray}
  \end{subequations}
  for joint distributions of the form
  \begin{equation}
    P_{XYZSU}(x,y,z,s,u) = P_S(s) \, P_{XU|S}(x,u|s) \, W(y,z|x,s).
  \end{equation}
\end{proposition}

\begin{proof}
  To bound $R_y$, we note that \eqref{eq:R1_3} holds also for the
  general broadcast channel \eqref{eq:general}, and we
  continue \eqref{eq:R1_3} as follows:
  \begin{IEEEeqnarray}{rCl}
    n R_y & \le & \sum_{i=1}^n I(M_y;Y_i|Y^{i-1},S^n) + n\epsilon_n\\
    & \le & \sum_{i=1}^n I(M_y, X_i ;Y_i|Y^{i-1},S^n) + n\epsilon_n\\
    & = & \sum_{i=1}^n H(Y_i|Y^{i-1},S^n) -H(Y_i|X_i,S_i) +
    n\epsilon_n \label{eq:R1_111} \IEEEeqnarraynumspace \\
    & \le & \sum_{i=1}^n I(X_i;Y_i|S_i) + n\epsilon_n. \label{eq:Ry_x}
  \end{IEEEeqnarray}
  Here \eqref{eq:R1_111} follows because, given $(X_i,S_i)$, the
  channel output~$Y_i$ is independent of
  $(M_y,Y^{i-1},S^{i-1},S_{i+1}^n)$. 

  We bound $R_z$ exactly as \eqref{eq:R2_7} with $V_i$,
  $i\in\{1,\ldots,n\}$, defined as in \eqref{eq:defV}.
  
  To bound the sum-rate $R_y+R_z$, note that \eqref{eq:R12_y},
  \eqref{eq:R12_z5}, and \eqref{eq:R12_36} still hold, but
  \eqref{eq:R12_37} should be replaced by
  \begin{equation}
    \sum_{i=1}^n I \left( \left. M_y,M_z,S^{i-1},S_{i+1}^n,
      Y_{i+1}^n; Y_i \right| S_i \right) = \sum_{i=1}^n
  I(X_i;Y_i|S_i),
  \end{equation}
  which is true because $(M_y,M_z,S^n)$ determines $X_i$, and because,
  without feedback,
  given $(X_i,S_i)$, the output $Y_i$ is independent of
  $(M_y,M_z,S^{i-1},S_{i+1}^n,Y_{i+1}^n)$. These together yield
  \begin{IEEEeqnarray}{rCl}
  n(R_y+R_z) & \le &  \sum_{i=1}^n I (V_i,T_i; Z_i)-\sum_{i=1}^n
  I(V_i,T_i;S_i,Y_i ) \nonumber \\    &&{}+
  \sum_{i=1}^n I(X_i;Y_i|S_i) + n \epsilon_n, \label{eq:Ryz_x}
  \end{IEEEeqnarray}
  where $T_i$, $i\in\{1,\ldots,n\}$, are defined in \eqref{eq:defT}.

  Summarizing \eqref{eq:Ry_x}, \eqref{eq:R2_7}, and \eqref{eq:Ryz_x} we
  conclude that the desired capacity region is contained in the convex
  closure of rate-pairs $(R_y,R_z)$ satisfying
  \begin{subequations}
  \begin{IEEEeqnarray}{rCl}
    R_y & < & I(X;Y|S)\\
    R_z & < & I(V;Z)-I(V;S)\\
    R_y+R_z & < & I(X;Y|S)+I(V,T;Z) - I(V,T;S,Y) \IEEEeqnarraynumspace
  \end{IEEEeqnarray}
  \end{subequations}
  where, given $(X,S)$, the outputs $(Y,Z)$ are drawn according to the
channel law \eqref{eq:general} independently of the auxiliary 
random variables $(V,T)$. Now to prove Proposition~\ref{prp:outer}
it remains to find a single auxiliary 
random variable $U$ satisfying
\begin{subequations}\label{eq:strange_general}
\begin{equation}
  I(V;Z) - I(V;S)  \le  I(U;Z) - I(U;S)
\end{equation}
and
\begin{IEEEeqnarray}{rCl}
  \lefteqn{I(X;Y|S) + I(V,T;Z) - I(V,T;S,Y)}~~~~~~~~~~~~~~~~
  \nonumber\\ 
  & \le & I(X;Y|S) + I(U;Z) - I(U;S,Y) \IEEEeqnarraynumspace
\end{IEEEeqnarray}
\end{subequations}
to replace both $V$ and $T$. Now note that \eqref{eq:strange_general}
is equivalent to~\eqref{eq:strange}. Hence, according to our arguments in
Section~\ref{sec:converse}, such a $U$ can always be found.
\end{proof}
 
\appendices

\section{Proof of Proposition~\ref{prp:cardinality}}\label{app:cardinality}
It suffices to show that, given any joint distribution $P_{XYZSU}$ of the form
\eqref{eq:distribution}, there exists another
distribution $\tilde{P}_{XYZSU}$ of the same form
  \begin{IEEEeqnarray}{rCl}
    \lefteqn{\tilde{P}_{XYZSU}(x,y, z,s,u)}~~~~~~~\nonumber \\ 
    & = & P_{S}(s) \, \tilde{P}_{XU|S}(x,u|s) \, \mathbf{1}\bigl\{y =
      f(x,s)\bigr\} \, W(z|x,s)\label{eq:sameform}
      \IEEEeqnarraynumspace 
  \end{IEEEeqnarray}
satisfying
\begin{equation}\label{eq:cardsmall}
  \left| \left\{ u\colon \tilde{P}_U(u) > 0 \right\} \right| \le
  |\set{X}| \cdot |\set{S}| + 1,
\end{equation}
where $\tilde{P}_U$ denotes the marginal of $\tilde{P}_{XYZSU}$ on
$U$, and
\begin{subequations}\label{eq:samebound}
\begin{IEEEeqnarray}{rCl}
  H(Y|S) \big\vert_{P} & = & H(Y|S)
  \big\vert_{\tilde{P}} \label{eq:samebound1} \\
  I(U;Z)-I(U;S)\big\vert_{P} & = &
  I(U;Z)-I(U;S)\big\vert_{\tilde{P}}   \IEEEeqnarraynumspace 
  \label{eq:samebound2} \\
  \lefteqn{H(Y|S)+I(U;Z)-I(U;S,Y)\big\vert_{P}}~~~~~~~ ~~~~~~~ ~~~~~~~
  ~~~~~~~ \nonumber\\
  \lefteqn{ = H(Y|S)+I(U;Z) -I(U;S,Y)\big\vert_{\tilde{P}}.}
  ~~~~~~~ 
  \IEEEeqnarraynumspace \label{eq:samebound3}
\end{IEEEeqnarray}
\end{subequations}

To this end, consider the following $|\set{X}|\cdot|\set{S}|+1$
functions of~$u$, all of which are determined by the conditional
distribution $P_{XYZS|U}$ and are independent of the marginal $P_U$:
\begin{subequations}
\begin{IEEEeqnarray}{rCl}
  h_0 (u) & \triangleq & H(S|U=u) - H(Z|U=u)\\
  h_1 (u) & \triangleq & H(Y,S|U=u) - H(Z|U=u)\\
  h_{x,s}(u) & \triangleq & P_{XS|U}(x,s|u),\nonumber \\
  & & ~~~~~~~x\in\set{X}, s\in\set{S}, (x,s)\neq (1,1).
\end{IEEEeqnarray}
\end{subequations}
We now look for a $\tilde{P}_U$ (which will replace $P_U$) such that
\begin{subequations}
\begin{IEEEeqnarray}{rCl}
  \sum_{u\in\set{U}} \tilde{P}_U(u) h_0 (u) & = &
  H(S|U)-H(Z|U)\big\vert_{P} \label{eq:constraint1}\\
  \sum_{u\in\set{U}} \tilde{P}_U(u) h_1 (u) & = &
  H(Y,S|U)-H(Z|U)\big\vert_{P} \label{eq:constraint2} \\
  \sum_{u\in\set{U}} \tilde{P}_U(u) h_{x,s} (u) & = &
  P_{XS}(x,s),\quad  \nonumber\\
  & & ~~~x\in\set{X},s\in\set{S}, (x,s)\neq
  (1,1). \IEEEeqnarraynumspace
  \label{eq:constraint3}
\end{IEEEeqnarray}
\end{subequations}
By the Support Lemma \cite[p.631]{elgamalkim11}, such a $\tilde{P}_U$
can be found whose support-size is at most the total number of constraints,
which equals $|\set{X}|\cdot|\set{S}|+1$. Choosing
\begin{equation}
  \tilde{P}_{XYZSU}(x,y,z,s,u) \triangleq \tilde{P}_U (u) \,
  P_{XYZS|U}(x,y,z,s|u) 
\end{equation}
for all $(x,y,z,s,u)$ yields a joint distribution that
satisfies~\eqref{eq:cardsmall}. We next show  
that this choice also satisfies \eqref{eq:sameform} and
\eqref{eq:samebound}. First note that \eqref{eq:constraint3} implies
that $\tilde{P}_{XYZUS}$ has the same marginal on $(X,S)$ as
$P_{XYZUS}$. In particular, 
\begin{equation}
  \tilde{P}_S(s) = P_S(s),\quad s\in\set{S}.
\end{equation}
This combined with the fact that we used the conditional distribution
$P_{XYZS|U}$ to generate $\tilde{P}_{XYZSU}$ shows that
$\tilde{P}_{XYZSU}$ is indeed of the form
\eqref{eq:sameform}. Furthermore, these imply that
\begin{equation}
  \tilde{P}_{XYZS} (x,y,z,s) = P_{XYZS}(x,y,z,s)
\end{equation}
for all $(x,y,z,s)$. Hence we have
\begin{subequations}\label{eq:yes1}
\begin{IEEEeqnarray}{rCl}
  H(Y|S) \big\vert_{\tilde{P}} & = & H(Y|S)\big\vert_{P}\\
  H(Z) - H(S) \big\vert_{\tilde{P}} & = & H(Z)-H(S) \big\vert_{P} \\
  H(Z) - H(Y,S) \big\vert_{\tilde{P}} & = & H(Z) - H(Y,S)
  \big\vert_{P}. \IEEEeqnarraynumspace
\end{IEEEeqnarray}
\end{subequations}
On the other hand, \eqref{eq:constraint1} and \eqref{eq:constraint2}
imply
\begin{subequations}\label{eq:yes2}
\begin{IEEEeqnarray}{rCl}
  H(S|U)-H(Z|U)\big\vert_{\tilde{P}} & = &
  H(S|U)-H(Z|U)\big\vert_{P}\\
  H(Y,S|U)-H(Z|U)\big\vert_{\tilde{P}} & = &
  H(Y,S|U)-H(Z|U)\big\vert_{P}. \IEEEeqnarraynumspace
\end{IEEEeqnarray}
\end{subequations}
Combining \eqref{eq:yes1} and \eqref{eq:yes2} yields
\eqref{eq:samebound} and concludes the proof.
  
\section{Proof of Proposition~\ref{prp:causal}} \label{app:example}


To prove Proposition~\ref{prp:causal}, we need the following simple
outer bound on the capacity region of any broadcast channel with
causal state-information:
\begin{lemma}\label{lem:causal}
  The capacity region of any state-dependent two-receiver broadcast
  channel as in \eqref{eq:general} with causal state-information at
  the transmitter is contained in the convex closure of the union of
  the rate pairs satisfying
  \begin{subequations}\label{eq:causal}
  \begin{IEEEeqnarray}{rCl}
    R_y & < & I(T;Y)\\
    R_z & < & I(T;Z)
  \end{IEEEeqnarray}
  \end{subequations}
  where the union is over all joint distributions of the form
  \begin{IEEEeqnarray}{rCl}
    \lefteqn{P_{XYZST}(x,y,z,s,t)}~~~~~~~~~~~~~\nonumber \\
    & = & P_S(s) \, P_T(t) \, \mathbf{1}\{x=g(t,s)\} \, W(y,z|x,s). \IEEEeqnarraynumspace
  \end{IEEEeqnarray}
  \end{lemma}

\begin{proof}
  We bound $R_y$ as for single-user channels with causal
  state-information \cite{shannon58,elgamalkim11} as follows:
  \begin{IEEEeqnarray}{rCl}
    nR_y & \le & I(M_y; Y^n) + n\epsilon_n\\
    & \le & I(M_y,M_z; Y^n) + n\epsilon_n\\
    & = & \sum_{i=1}^n I(M_y,M_z;Y_i|Y^{i-1}) + n\epsilon_n\\
    & \le & \sum_{i=1}^n I(M_y,M_z,Y^{i-1};Y_i)  + n\epsilon_n\\
    & \le & \sum_{i=1}^n I(M_y,M_z,S^{i-1},Y^{i-1};Y_i)  +
    n\epsilon_n\\
    & = & \sum_{i=1}^n I(M_y,M_z,S^{i-1},X^{i-1},Y^{i-1};Y_i)  +
    n\epsilon_n \IEEEeqnarraynumspace \label{eq:causal1}\\
    & = & \sum_{i=1}^n I(M_y,M_z,S^{i-1}, X^{i-1};Y_i)  +
    n\epsilon_n \label{eq:causal2} \\
    & = & \sum_{i=1}^n I(M_y,M_z,S^{i-1};Y_i)  +
    n\epsilon_n. \label{eq:causal3} 
  \end{IEEEeqnarray}
  Here, \eqref{eq:causal1} and \eqref{eq:causal3} follow because
  $X^{i-1}$ is a function of $(M_y,M_z,S^{i-1})$; and
  \eqref{eq:causal2} because, given $(M_y,M_z,S^{i-1},X^{i-1})$, the
  output $Y_i$ is independent of $Y^{i-1}$. In the same way we can
  obtain
  \begin{equation}
    nR_z \le \sum_{i=1}^n I(M_y,M_z,S^{i-1};Z_i).\label{eq:causal4}
  \end{equation}
  We define
  \begin{equation}
    T_i \triangleq (M_y,M_z,S^{i-1}),\quad i \in\{1,\ldots,n\}
  \end{equation}
  which clearly satisfy the conditions
  \begin{equation}
    T_i\indep S_i,\quad T_i\markov (X_i,S_i) \markov (Y_i,Z_i),\quad
    i\in \{1,\ldots,n\}. 
  \end{equation}
  We now have
  \begin{subequations}
  \begin{IEEEeqnarray}{rCl}
    nR_y & \le & \sum_{i=1}^n I(T_i;Y_i) + n\epsilon_n\\
    nR_z & \le & \sum_{i=1}^n I(T_i;Z_i) + n\epsilon_n,
  \end{IEEEeqnarray}
  \end{subequations}
  which imply that the capacity region of interest is contained in
  the convex closure of \eqref{eq:causal} for distributions on
  $(X,Y,Z,S,T)$ satisfying
  \begin{equation}
    T\indep S,\quad T \markov (X,S) \markov (Y,Z).
  \end{equation}
  It now only remains to show that, to exhaust this region,
  it suffices to consider joint distributions in which $X$ is a function
  of $(T,S)$. This is indeed the case because, given $P_{TS}(t,s)$ and
  the channel law, both terms on the RHS of \eqref{eq:causal} are
  convex in~$P_{X|TS}$.
\end{proof}

  We next proceed to prove Proposition~\ref{prp:causal}. We begin with
  the achievability part, which is 
  straightforward. If the transmitter only communicates to the
  receiver which observes $Y$, then it can cancel the interference of
  $S$ by flipping the input symbol whenever $S=1$. In this way the
  rate-pair
  \begin{equation}\label{eq:share1}
    (R_y,R_z) = (1,0)
  \end{equation}
  can be achieved. On the other hand, if the transmitter only
  communicates to the receiver which observes $Z$, then it can ignore
  $S$ and achieve the rate-pair
  \begin{equation}\label{eq:share2}
    (R_y,R_z) = (0, 1-\Hb(p)).
  \end{equation}
  Time-sharing between \eqref{eq:share1} and \eqref{eq:share2}
  achieves the claimed capacity region.

  To prove the converse part, we use Lemma~\ref{lem:causal}. Note that
  the auxiliary random variable $T$ in Lemma~\ref{lem:causal} can be
  restricted to take value in all ``input strategies''
  \cite{shannon58}. Namely, its alphabet is the set of all mappings
  from $\set{S}$ to $\set{X}$. There are four such mappings:
  \begin{subequations}
  \begin{IEEEeqnarray}{rCl}
    T & = & 0\colon \quad \textnormal{maps $0$ to $0$ and $1$ to $0$}\\
    T & = & 1\colon \quad \textnormal{maps $0$ to $1$ and $1$ to $1$}\\
    T & = & 2\colon \quad \textnormal{maps $0$ to $0$ and $1$ to $1$}\\
    T & = & 3\colon \quad \textnormal{maps $0$ to $1$ and $1$ to $0$}.
  \end{IEEEeqnarray}
  \end{subequations}
  Here, $T=0$~or~$1$ means sending a fixed $x$ independently of $S$,
  and $T=2$~or~$3$ means flipping $x$ whenever $S=1$. Using the
  ``fixed'' strategies $T=0$ or $1$ one can transmit 
  information to 
  the receiver which observes $Z$ but not to the receiver which
  observes $Y$:
  \begin{IEEEeqnarray}{rCl}
    H(Y|T=0) & = & H(Y|T=1) = 1\\
    H(Z|T=0) & = & H(Z|T=1) = 1-\Hb(p);
  \end{IEEEeqnarray}
  while using the ``flipped'' strategies $T=2$ or $3$ one can transmit
  information 
  to the receiver which observes $Y$ but not to the receiver which
  observes $Z$:
  \begin{IEEEeqnarray}{rCl}
    H(Y|T=2) & = & H(Y|T=3) = 0\\
    H(Z|T=2) & = & H(Z|T=3) = 1.
  \end{IEEEeqnarray}
  We now have
  \begin{IEEEeqnarray}{rCl}
    R_y & \le & I(T;Y)\\
    & = & H(Y) - H(Y|T)\\
    & = & H(Y) - P_T(0) H(Y|T=0) \nonumber\\
    & & {} - P_T(1) H(Y|T=1)\\
    & \le & 1 - \Pr \bigl[ T\in\{0,1\} \bigr]\cdot 1\\
    & = & \Pr \bigl[ T\in\{2,3\} \bigr]\\
    R_z & \le & I(T;Z)\\
    & = & H(Z) - H(Z|T)\\
    & = & H(Z) - P_T(0) H(Z|T=0) - P_T(1) H(Z|T=1) \nonumber\\
    & & {} - P_T(2)H(Z|T=2) -
    P_T(3)H(Z|T=3) \\
    & \le & 1 - \Pr \bigl[ T\in\{0,1\} \bigr]\cdot (1-\Hb(p))
    \nonumber\\
    & & {} - \Pr
    \bigl[ T\in\{2,3\} \bigr]\cdot 1\\
    & = & \left(1- \Pr \bigl[ T\in\{2,3\} \bigr] \right) \cdot
    (1-\Hb(p)). 
  \end{IEEEeqnarray}
  Denoting 
  \begin{equation}
    \lambda \triangleq \Pr \bigl[ T\in\{2,3\} \bigr]
  \end{equation}
  we see that $(R_y,R_z)$ indeed must satisfy \eqref{eq:causal}.
  This ends our proof of Proposition~\ref{prp:causal}.

\section*{Acknowledgments}

The authors thank the anonymous reviewers of both the conference and
the journal versions of this paper, as well as the Associate Editor
Yossef Steinberg for their helpful comments. 


\bibliographystyle{IEEEtran}           

\end{document}